%% file: two_graph.tex
\documentclass{article}
\usepackage[utf8]{inputenc}
\usepackage{amssymb,amsfonts,amsmath,amsthm}
\usepackage[margin=1in]{geometry}
\usepackage{graphicx}
\usepackage{caption}
\usepackage{subcaption}
\usepackage{xcolor}
\usepackage[capitalise]{cleveref}
\usepackage{url}

\usepackage{authblk}
\usepackage{paralist}

\newtheorem{theorem}{Theorem}

\theoremstyle{definition}

\def\E{\mathbb{E}}
\def\Pr{\mathbb{P}}

\def\deg{\operatorname{deg}}
\def\Sq{\operatorname{Sq}}

\def\Cir{\operatorname{Cir}}

\def\SI{the Appendix}
\def\thmminusedge{Theorem~1}
\def\thmbipartite{Theorem~2}
\def\thmoned{Theorem~3}

\def\pan{P}
\def\ttop{TT}
\def\lol{LP}

\def\fii{\varphi}
\def\fp{\rho}
\def\fpx{\fp^\star}
\def\fponed{\fp^{\operatorname{1D}}_N}
\def\fptwod{ \fp^{\operatorname{2D}}_N}

\def\rtwod{ r^{\operatorname{2D}}_N}
\def\fpknr{\fp(K_N;r)}
\def\Mhigh{M^\uparrow}
\def\Mlow{M^\downarrow}

\begin{document}

\author[1]{Josef Tkadlec}
\author[2]{Kamran Kaveh} 
\author[3]{Krishnendu Chatterjee} 
\author[1,4]{Martin A. Nowak}

\affil[1]{Department of Mathematics, Harvard University, Cambridge, MA 02138, USA}
\affil[2]{Department of Applied Mathematics, University of Washington, Seattle, WA 98195, USA}
\affil[3]{Institute of Science and Technology Austria, Am Campus 1, 3400 Klosterneuburg, Austria}
\affil[4]{Department of Organismic and Evolutionary Biology, Harvard University, Cambridge, MA 02138, USA}
\affil[$\star$]{To whom correspondence should be addressed: tkadlec@math.harvard.edu}

\title{Natural selection of mutants that modify population structure}
\date{}

\maketitle

\begin{abstract}
Evolution occurs in populations of reproducing individuals.
It is well known that population structure can affect evolutionary dynamics.
Traditionally, natural selection is studied between mutants that differ in reproductive rate,
but are subject to the same population structure.
Here we study how natural selection acts on mutants that have the same reproductive rate,
but experience different population structures.
In our framework, mutation alters population structure, which is given by a graph that
specifies the dispersal of offspring.
Reproduction can be either genetic or cultural.
Competing mutants disperse their offspring on different graphs.
A more connected graph implies higher motility. 
We show that enhanced motility tends to increase an invader’s fixation probability, but there are interesting exceptions.
For island models, we show that the magnitude of the effect depends crucially on the exact layout of the additional links.
Finally, we show that for low-dimensional lattices, the effect of altered motility is comparable to that of altered fitness:
in the limit of large population size, the invader's fixation probability is
either constant or exponentially small, depending on whether it is more or less motile than the resident.
\end{abstract}

\section*{Introduction}

Evolutionary dynamics is the study of how different traits arise and disappear in a population of reproducing individuals.
Each trait might confer a fitness advantage (or disadvantage) on its bearer,
thus in turn altering the probability that the trait spreads through the population (an event called \textit{fixation}) or disappears (\textit{extinction}).
Besides the fitness advantage, another important factor in determining the fate of a trait over time (its fixation or extinction)
is the spatial structure of the population \cite{durrett2008probability,nowak2006evolutionary,broom2014game,moran,nagylaki1992introduction}.
For instance, the population might be subdivided into ``islands'': 
An offspring of a reproducing individual then typically stays in the same island, but occasionally it migrates to some nearby island.
The fixation probability of a trait then crucially depends on the dispersal pattern, that is, the migration rates among the islands.
Incorporation of population structure into a model of selection dynamics substantially improves the descriptive power of the model ~\cite{durrett2008probability,nagylaki1992introduction,pollak1966survival,nagylaki1980strong,whitlock1997effective,durrett1994importance,komarova2006spatial,santos2006evolutionary}.

Evolutionary graph theory is a powerful framework for studying natural selection in population structures with arbitrarily complex dispersal patterns~\cite{lieberman2005evolutionary,antal2006evolutionary,broom2008analysis,diaz2014approximating,adlam2015amplifiers,monk2018martingales,allen2017evolutionary}.
On an evolutionary graph (network), individuals occupy the nodes (vertices), and the edges (links) specify where the offspring can migrate.
Graphs can represent spatial structures, contact networks in epidemiology, social networks, and phenotypic or genotypic structures in biological populations~\cite{lieberman2005evolutionary, santos2005scale, keeling2011modeling, szabo2007evolutionary, castellano2009statistical,perc2013evolutionary}.
The question is then: How does a graph structure affect the fixation probability of a new mutant introduced into a background population of residents?
Extensive research over the past decade has produced many remarkable population structures with various desirable properties~\cite{broom2011stars,mertzios2013natural,Galanis17,tkadlec2019population,allen2021fixation}.
As one example, consider a mutation that increases the reproduction rate of the affected individual.
Population structures that increase the fixation probability of such mutations, as compared to the baseline case of unstructured (well-mixed) populations, are known as amplifiers of selection.
Many amplifiers of selection are known, both simple ones and strong ones~\cite{monk2014martingales,pavlogiannis2018construction,Goldberg19,tkadlec2021fast}.

In this work, we consider mutations that do not change the reproductive rate of the affected individual, but rather its motility potential.
In nature, an altered motility potential could arise in a variety of scenarios.
We give three examples.

First, consider a species occupying a region that is split by a geographical barrier into two parts.
If the mutation allows the offspring to successfully cross the barrier, the 
mutants will perceive the population structure as being close to well-mixed, whereas
the residents will continue perceiving it as being split into two parts (islands).

As a second example, consider structured multicellular organisms.
There, cells are arranged in symmetric lattice structures known as epithelia.
An epithelial tissue may be described as a two-dimensional sheet defined by vertex points representing wall junctions,
one-dimensional edges representing cell walls, and two-dimensional faces representing cells.
The form of this tissue network is determined by the extracellular matrix (ECM).
The ECM is a network consisting of
extracellular macromolecules, collagen, and enzymes that provide structural and biochemical support to surrounding cells.
The composition of ECM varies between multicellular structures \cite{frantz2010extracellular,hay2013cell,walker2018role,gibson2009cell,kachalo2015mechanical}.
Thus, when discussing somatic evolution in multicellular organisms,
the invading genotype might differ in what network structure it is forming \cite{radisky2002order,walker2018role}.
In other words, each type, in the absence of the other type, forms its own and different extracellular matrix.
This leads to different alignment of cells and thus a new population structure, see~\cref{fig:tissue_1}.

\begin{figure}
\begin{center}
\includegraphics[scale=1]{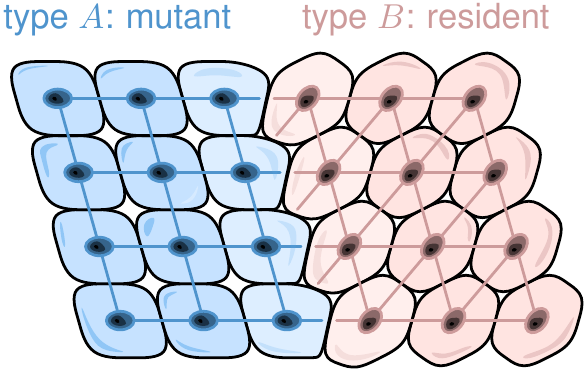}
\end{center}
\caption{In epithelial tissues, different cell types align along different lattice-like structures.}
\label{fig:tissue_1}
\end{figure}

Carcinoma is yet another example of how the tissue organization of the invader and resident type can differ from each other.
In this case, tumor cells normally have a highly disorganized neighborhood structure, due to
the variability in cell-cell adhesion and the lack of proper epithelial programs among tumor cells in the tumor microenvironment \cite{nelson2006extracellular,brauchle2018biomechanical}.
Normal epithelial cells, on the other hand, typically follow symmetric geometric lattice patterns.
This change in structure between an invading trait and the resident type
can have substantial consequences on the outcome of the evolutionary process.
However, in the context of evolutionary graph models, such considerations have not yet received appropriate attention. 

In order to model differences in the motility potential within the framework of evolutionary graph theory,
we represent the population structure as two graphs $G^A$, $G^B$
overlaid on top of each other on the same set of nodes~\cite{spirakis2021extension}.
The two graphs $G^A$, $G^B$ represent the dispersal patterns for the mutants and residents, respectively.
In other words, mutant offspring migrate along the edges of $G^A$,
whereas resident offspring migrate along the edges of $G^B$.
We study the fixation probability $\fp(G^A,G^B)$ of a single neutral mutant who
appears at a random node and perceives the population structure as $G^A$,
as it attempts to invade a population of residents who perceive the population through $G^B$.

There is a large body of literature on the evolution and ecology of migration and dispersal~\cite{comins1980evolutionarily,dieckmann1999evolutionary,hutson2003evolution,levin2003ecology,ronce2007does},
especially for population structures formed by islands (also called patches, demes, or metapopulations)~\cite{may1994superinfection,olivieri1995metapopulation,heino2001evolution}.
Our framework is a generalization of this approach in the same way that evolutionary graph theory is a generalization of the
vast literature on evolution and ecology in spatially structured populations~\cite{lieberman2005evolutionary,durrett1994importance}.
The framework is flexible, allowing us to study both simple and arbitrarily complex population structures of any population size.
As such, it facilitates a discovery of new phenomena.

Among the graph-theoretical approaches, other ways to model motility and dispersal have been suggested in the literature.
They allow for the offsprings to disperse in more complex forms and reach locations that are not directly connected to the mother location.
This introduces migration potential as an independent quantity relative to the proliferation potential of the types \cite{ohtsuki2007evolutionary,thalhauser2010selection,manem2014spatial,krieger2017effects,herrerias2019motion,waclaw2015spatial,manem2015modeling}.
In those cases, the motility potential is representative of a random motion and it is typically decoupled from the reproduction events.
Such random motility and motion has an anti-synergistic relationship with the proliferation potential.
In other words, if invaders are more motile, their fixation probability tends to decrease \cite{thalhauser2010selection,manem2014spatial,krieger2017effects}.

Here we show that, in contrast to random motility, enhanced structured motility generally leads to an increase in the fixation probability of the invading mutant.
Specifically, we prove that for any population size $N$ the Complete graph $K_N$ is ``locally optimal''.
That is, if mutants instead perceive the population through a graph $M_N$ that misses a single edge, 
their fixation probability is decreased.
However, we show that the obvious generalization of this claim is not true:
By numerically computing the fixation probabilities for small population sizes,
we identify specific circumstances in which 
making mutants less motile actually increases their fixation probability.
Next, we show that even for simple population structures that correspond to island models,
the extent to which increased motility helps the mutant fixate can vary considerably, depending on the exact layout of the extra connections.
Finally, we show that for low-dimensional lattices,
the effect of altered motility is comparable to the effect of altered reproductive rate:
in the limit of large population size,
the fixation probability of a mutant is either constant or exponentially small, depending on whether it is more or less motile than the residents.

\section{Model}
\paragraph{Standard Moran process on a graph.}
Within the framework of Evolutionary graph theory~\cite{lieberman2005evolutionary},
a population structure is described as a graph (network), where
nodes (vertices) represent locations (sites) and the graph connectivity defines the topology and the neighborhood.
There are $N$ nodes and each node is occupied by a single individual.
Each individual is either of type~$A$ (mutant) with fitness $r_A$, or of type~$B$ (resident) with fitness $r_B$.
The evolutionary dynamics is governed by the standard stochastic discrete-time Moran Birth-death process, adapted to the population structure:
at each time point, a single individual is picked for reproduction, proportionally to its fitness.
This focal individual produces offspring (a copy of itself), and the offspring then migrates and replaces a random neighboring individual.

The probability of migration from node $i$ to node $j$ is given by an $N\times N$ dispersal matrix $M=(m_{i,j})_{i,j=1}^N$.
Thus, for undirected, unweighted graphs (which are the focus of this work), the entries $m_{i,j}$ of the dispersal matrix $M$ satisfy
\[m_{i,j}=\begin{cases}
 1/\deg(i), &\text{ if nodes $i$ and $j$ are adjacent,}\\
 0, &\text{ otherwise.}
 \end{cases}
 \]
 (Here $\deg(u)$ is the \textit{degree} of node $u$, that is, the number of nodes adjacent to $u$.)

\paragraph{Moran process on two graphs.}
It is commonly assumed that the dispersal matrix is independent of the two types, that is,
both types of individuals perceive the population through the same population structure.
Following the recent work of Melissourgos et al.~\cite{spirakis2021extension}, here we study a more general case in which
the dispersal pattern depends on the type of the offspring that migrates.
Thus, we consider two graphs $G^A$, $G^B$ and the corresponding dispersal matrices $M^A=(m^A_{i,j})_{i,j=1}^N$, $M^B=(m^B_{i,j})_{i,j=1}^N$.
That is, any time a type~$A$ individual reproduces at a node $i$, the offspring replaces an individual at node $j$ with probability $m^A_{ij}$.
In contrast, the offspring of a type~$B$ individual reproducing at node $i$ migrates to node $j$ with probability $m^B_{ij}$, see~\cref{fig:fig1}.

\begin{figure}[h]
\begin{center}
\includegraphics[scale=1]{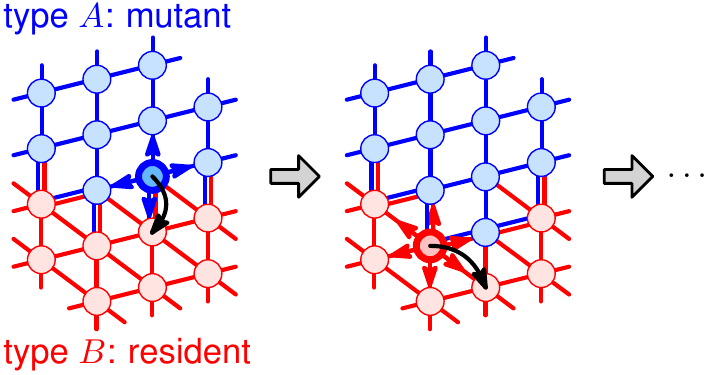}
\end{center}
\caption{\textbf{Moran process with type-dependent dispersal patterns.}
In each discrete time-step, a random individual reproduces and the offspring proliferates to a neighboring node.
Type-$A$ offspring (mutant, blue) migrate along the edges of the blue graph $G_A$,
whereas type-$B$ offspring (residents, red) migrate along the red edges of $G_B$.
The key quantity is the fixation probability $\fp(G^A,G^B)$ that a single initial mutant 
successfully invades the population of residents. 
}
\label{fig:fig1}
\end{figure}

The state of the population at any given time point is described by a vector ${\bf n}=(n_1,\dots,n_N)$ of $N$ zeros and ones, where $n_i=1$ denotes that node $i$ is currently occupied by a type~$A$ individual (mutant).
The model is a Markov chain with $2^N$ possible states.
Two of the states are absorbing, and they correspond to homogeneous population consisting purely of type~$A$ individuals (state ${\bf n^1}=(1,\dots,1)$) or type~$B$ individuals (state ${\bf n^0}=(0,\dots,0)$).
Formally, the transition probabilities between the states are given by the following equations:

\begin{align}
p^{+}_{i}({\bf n}) :=& \Pr[ (n_1,\dots,n_i,\dots, n_N)  \ \to\  (n_1,\dots,n_i+1,\dots, n_N)] \nonumber\\
=& \frac{\sum_{j} n_j(1-n_i) r_A m^{A}_{ji}}{\sum_{k}\left(n_kr_A+ (1-n_k)r_B\right)  }  \nonumber\\
p^{-}_{i}({\bf n}) :=& \Pr[(n_{1},\dots,n_{i},\dots, n_{N})  \ \to\  (n_{1},\dots,n_{i}-1,\dots, n_{N})] \nonumber\\
=&\frac{\sum_{j} (1-n_{j})n_{i} r_{B}m^{B}_{ji}}{\sum_{k}\left(n_kr_A+ (1-n_k)r_B\right)}
\label{transition}
\end{align}

\paragraph{Questions and Results.}
In this work, we study how differences in the migration and dispersal pattern
$G^A$ of mutants and $G^B$ of residents influence
the fate of a single random mutant who appears at a random location.
As a measure of the mutant success, we use its fixation probability under neutral drift (that is, $r_A=r_B$).
We denote this quantity by $\fp(G^A, G^B)$.
It is known that whenever the two types have the same dispersal pattern ($G^A=G^B$), the fixation probability under neutral drift is equal to $1/N$, regardless of $G^A=G^B$~\cite{broom2010two}.
Thus, the regime of neutral drift provides a clean baseline and it decouples the effect of a difference in population structure from other effects. 

Specifically, we study the following questions:
\begin{enumerate}
\item Does increased motility increase or decrease the mutant fixation probability?
\item Can the effect be quantified for simple natural structures, such as island models or low-dimensional lattices?
\end{enumerate}

To address the first question,
in~\cref{sec:small} we numerically compute the fixation probabilities $\fp(G^A, G^B)$ for all pairs $G^A$, $G^B$ of graphs of small size.
We find that, generally speaking, increased motility potential (that is, living on a graph with more edges) tends to
increase the fixation probability of the mutant.
In particular, we prove (see~\thmminusedge{} in~\SI) that the Complete graph is locally optimal, in a sense described below.
However, we also identify special cases, in which
an increase in the motility potential decreases the fixation probability rather than increasing it.
This suggests that for arbitrary population structures the effects of motility on the fixation probability are complex.
Given this complexity, we proceed to study pairs of regular structures.

To address the second question,
in~\cref{sec:dense} we consider certain population structures that correspond to island models with two equal islands.
We show that two such structures with the same total number of edges exhibit a substantially different behavior in the limit $N\to\infty$.
This implies that the effect of altered motility in dense regular graphs can not be easily quantified in terms of a single parameter (the total number of edges).
Then, motivated by tissue organization in multicellular organisms, in~\cref{sec:lattices} we consider 1- and 2-dimensional lattices.
We show that in this setting, the difference in motility can be quantified and it has analogous effect to a difference in reproductive rate:
increased motility results in mutant fixation with constant probability, whereas decreased motility causes the fixation probability to be exponentially small.

\paragraph{Related work.}
The question of computing fixation probabilities for various versions of Moran processes on graphs has been studied extensively.
In principle, for any population structure the fixation probability can be computed numerically by solving a system of linear equations~\cite{hindersin2016exact}.
However since the size of the system is generally exponential in the population size, this approach is practically feasible only for very small populations, or for very specific population structures~\cite{tkadlec2019population,moller2019exploring,pavlogiannis2017amplification}.
For large population sizes, there exist efficient approximation algorithms either in the limit of weak selection~\cite{allen2017evolutionary,allen2021fixation,mcavoy2021fixation} 
or when the underlying graph is undirected~\cite{diaz2014approximating,ann2020phase}. 
While this manuscript was under preperation, Melissourgos et al.~\cite{spirakis2021extension}
extended the latter result to a special case of the two-graph setting, namely
for mutants with reproductive advantage  ($r_A\gneqq r_B$) who perceive the population as a Complete graph ($G_A=K_N$).
They also established bounds for certain special pairs of graphs, such as
the Complete graph invading the Star graph. 
In contrast, in this work we consider the problem from a biological perspective and
we study mutants with no reproductive advantage ($r_A=r_B$) who, similarly to the residents,
perceive the population structure either as an island model or as a low-dimensional lattice.
In this way, the two manuscripts complement each other.
We also answer some questions stated in~\cite{spirakis2021extension} related to the best-response dynamics in the space of all graphs.
Namely, we show that while the Complete graph is locally optimal (see~\thmminusedge{} in~\SI), it is not always the best response  (see~\cref{fig:game}).

\section{Results}

\subsection{Small Graphs}\label{sec:small}
In this section we consider population structures on $N$ labelled nodes, for small values of $N$.
In this regime, the fixation probability $\fp(G^A,G^B)$ can be computed exactly,
by numerically solving a system of $2^N$ linear equations. 

For $N=2$ there is only one connected graph and, by symmetry, the fixation probability of a single type~$A$ individual is equal to 1/2.
For $N=3$ there are four undirected graphs: a single graph $G^0$ with three edges (equivalently a complete graph, or a cycle), and
three different graphs $G^1$, $G^2$, $G^3$ with two edges each.
The corresponding fixation probabilities are given in~\cref{fig:n3}b.
Note that $\fp(G^A,G^B)=1/N$ when $G^A$ and $G^B$ are identical, but in general $\fp(G^A,G^B)$ could be both more than $1/N$ or less than $1/N$, even when $G^A$ and $G^B$ are isomorphic (if they are not identical), see~\cref{fig:n3}c.

\begin{figure}[h]
\begin{center}
\includegraphics[scale=1]{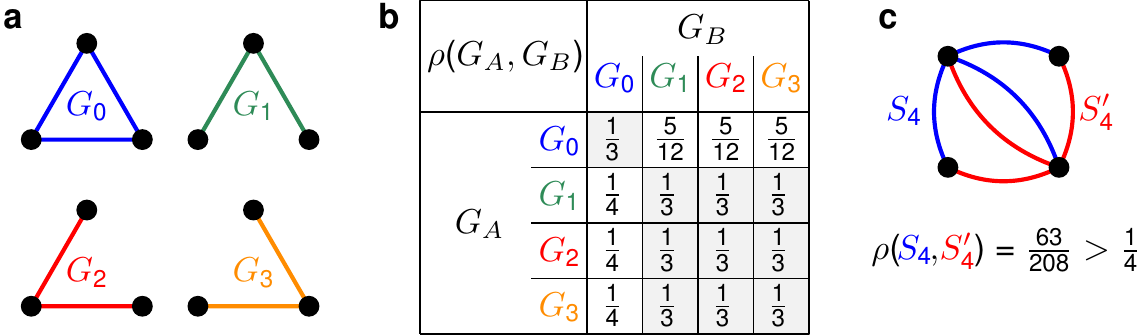}
\end{center}
\caption{\textbf{Small populations $N=3$.}
\textbf{a,} There are four connected graphs $G^0,\dots,G^3$ on $N=3$ labeled nodes.
\textbf{b,} The fixation probabilities $\fp(G^A,G^B)$ for all $4\cdot 4=16$ combinations.
\textbf{c,} When $G^A$ and $G^B$ are isomorphic but not identical, the fixation probability is not necessarily equal to $1/N$. For instance,
 we have $\fp(S_4,S'_4)=63/208\doteq 0.31$. 
}
\label{fig:n3}
\end{figure}
For general $N$, there are $2^{N^2-N}$ pairs of graphs on $N$ labeled nodes.
Already for $N=6$ this is more than a billion pairs, hence
in what follows we focus on the case when one of the graphs $G^A$, $G^B$ is a Complete graph, denoted $K_N$.
We use a shorthand notation $\fp(G)=\fp(G,K_N)$, for the fixation probability of a single mutant who perceives the population structure as a graph $G$ and invades a population of residents who perceive the population structure as a Complete graph $K_N$.
Analogously, we denote by $\fpx(G)= \fp(K_N,G)$ the fixation probability of a single mutant living on a Complete graph $K_N$ and invading a population of residents who live on~$G$.
\cref{fig:n6} shows $\fp(G)$ and $\fpx(G)$ for all undirected graphs on $N=6$ vertices, based on the number of edges in $G$.

\begin{figure}[h] 
	\centering
\includegraphics[scale=1]{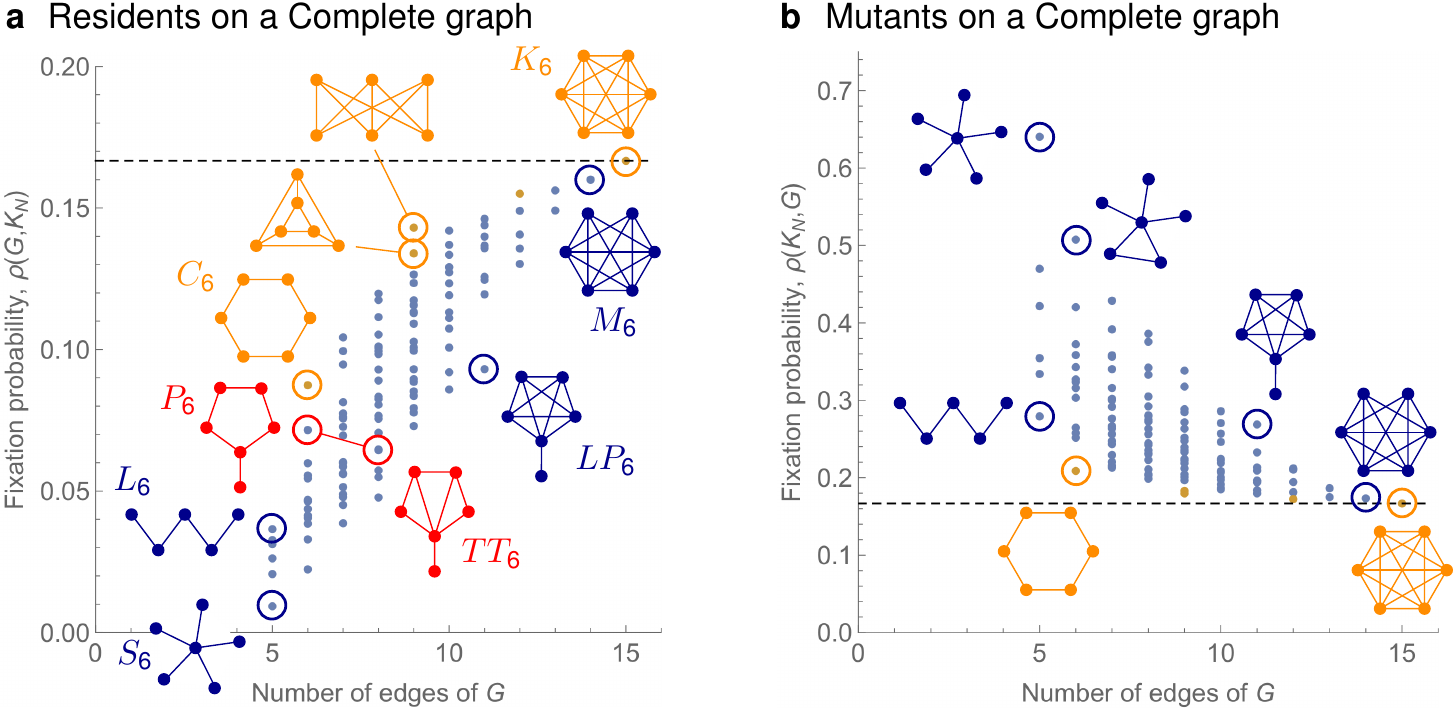}
        \caption{\textbf{Small populations $N=6$.} The fixation probabilities
        \textbf{a,} $\fp(G)=\fp(G,K_N)$ and
        \textbf{b,} $\fpx(G)=\fp(K_N,G)$ 
        for all 112 graphs $G$ on $N=6$ vertices.
        Each dot corresponds to a graph~$G$, the orange dots correspond to regular graphs.
        When $G=K_N$, both $\fp(G)$ and $\fpx(G)$ are equal to $1/N$.
        Other graphs $G_6$ on six vertices satisfy $\fp(G_6)<1/6$ and $\fpx(G_6)>1/6$.       
}
\label{fig:n6}
\end{figure}

\paragraph{Maximal and minimal fixation probability.} 
Among the graphs on 6 vertices, fixation probability $\fp(G)$ is maximized when $G$ is the Complete graph $K_6$.
Recall that $\fp(K_N)=1/N$, for any integer $N$.
In relation to this, we prove that $\fp(K_N)$ is ``locally maximal'':
that is, we show that if one edge is removed from the Complete graph $K_N$, then the resulting graph $M_N$ satisfies
$\fp(M_N)=\frac{N-2}{(N-1)^2}<\frac1N=\fp(K_N)$.
Similarly, we prove that $K_N$ is locally minimal with respect to $\fpx(G)$:
we show that $\fpx(M_N)=1/(N-1)>1/N$, see~\thmminusedge{} in \SI.

Note that, in contrast, for $N=6$ the fixation probability $\fp(G)$ is minimized for the Star graph $S_6$.
Here a \textit{Star graph}, denoted $S_N$, consists of a single node (``center'') connected to all other nodes (``leaves'').
It is known~\cite{spirakis2021extension} that $\fp(S_N)\le 1/(N-2)!$ and $\fpx(S_N)\to 1$ as $N\to\infty$.

\paragraph{Relation to the number of edges.}
In general, fixation probability $\fp(G)$ tends to be higher for graphs $G$ with more edges.
However, this is only a rule of thumb.
For instance, the Lollipop graph $\lol_6$ has a relatively low fixation probability $\fp(\lol_6)$, given its number of edges.
Here a \textit{Lollipop graph}, denoted $\lol_N$, consists of a Complete graph on $N-1$ vertices and a single extra edge connecting the last node.
Moreover, adding edges to a graph $G$ to produce a graph $G'$ sometimes does not increase the fixation probability but rather decreases it:
this is illustrated by the Pan graph $\pan_6$ and the Treetop graph $\ttop_6$ for which we have $\fp(\pan_6)>0.071$ and $\fp(\ttop_6)<0.065$.
Here a \textit{Pan graph}, denoted $\pan_N$, consists of a cycle on $N-1$ nodes and a single extra edge connecting the last node.
In a \textit{Treetop graph}, denoted $\ttop_N$, the vertex with degree 3 is further connected to all other vertices.

\paragraph{Regular graphs.}
Recall that a graph is \textit{regular} if all its nodes have the same degree (that is, the same number of neighbors).
\cref{fig:n6} shows that, given a fixed number of edges, the fixation probability $\fp(G)$ tends to be higher for regular (or almost regular) graphs as compared to non-regular graphs.
For instance, for the Cycle graph $C_6$ and the Line graph $L_6$, the fixation probabilities $\fp(C_6)$, $\fp(L_6)$ are relatively high, given the low number of edges of $C_6$ and $L_6$.
Here a \textit{Cycle graph}, denoted $C_N$, is the connected graph where each node is connected to two neighbors, and
a \textit{Line graph}, denoted $L_N$, is the Cycle graph with one edge missing.
However, we prove that the Line graph generally does not maximize the fixation probability among the connected graphs with $N-1$ edges (so-called trees):
in particular, for $N=8$ the graph $G_8$ consisting of three paths of lengths 2, 2, and 3 meeting at a single vertex satisfies $\fp(G_8)>0.0098>0.0095>\fp(L_8)$.

Moreover, the Isothermal Theorem of~\cite{lieberman2005evolutionary}
 does not hold:
for two different regular graphs $G$, $G'$ (even with the same degree) the fixation probabilities $\fp(G)$, $\fp(G')$ are generally different, as witnessed by the two 3-regular graphs with $N=6$ nodes and 9 edges.

%
%
\subsection{Dense regular graphs}\label{sec:dense}

As suggested by~\cref{fig:n6}, regular graphs $G$ have high fixation probability $\fp(G)$, compared to other graphs with the same number of edges.
Here we consider certain simple regular graphs that contain approximately half of the total possible number of edges.
We show that for some such graphs, the fixation probability is comparable to that of a Complete graph,
whereas for other graphs it is substantially smaller. Thus the Isothermal Theorem~\cite{lieberman2005evolutionary} is strongly violated.

Given a population size $N$ (with $N$ even), let $B_N=K_{N/2,N/2}$ be a (complete) \textit{Bipartite graph} with equal parts $N/2$, $N/2$ and let $T_N$ be a \textit{Two-clique} graph obtained by adding $N/2$ matching edges to a union of two disjoint Complete graphs of size $N/2$ each, see~\cref{fig:bntn}a.
Note that both $B_N$ and $T_N$ have precisely $\frac14N^2$ edges, which is roughly half of the edges of $K_N$.
Also, note that both $B_N$ and $T_N$ represent populations subdivided into two large islands:
in case of $B_N$, the offspring always migrates to the opposite island, whereas
in case of $T_N$ the offspring mostly stays in the same island and it migrates only rarely (namely with probability of the order of $1/N$).

\begin{figure}[h] 
	\centering
\includegraphics[width=\linewidth]{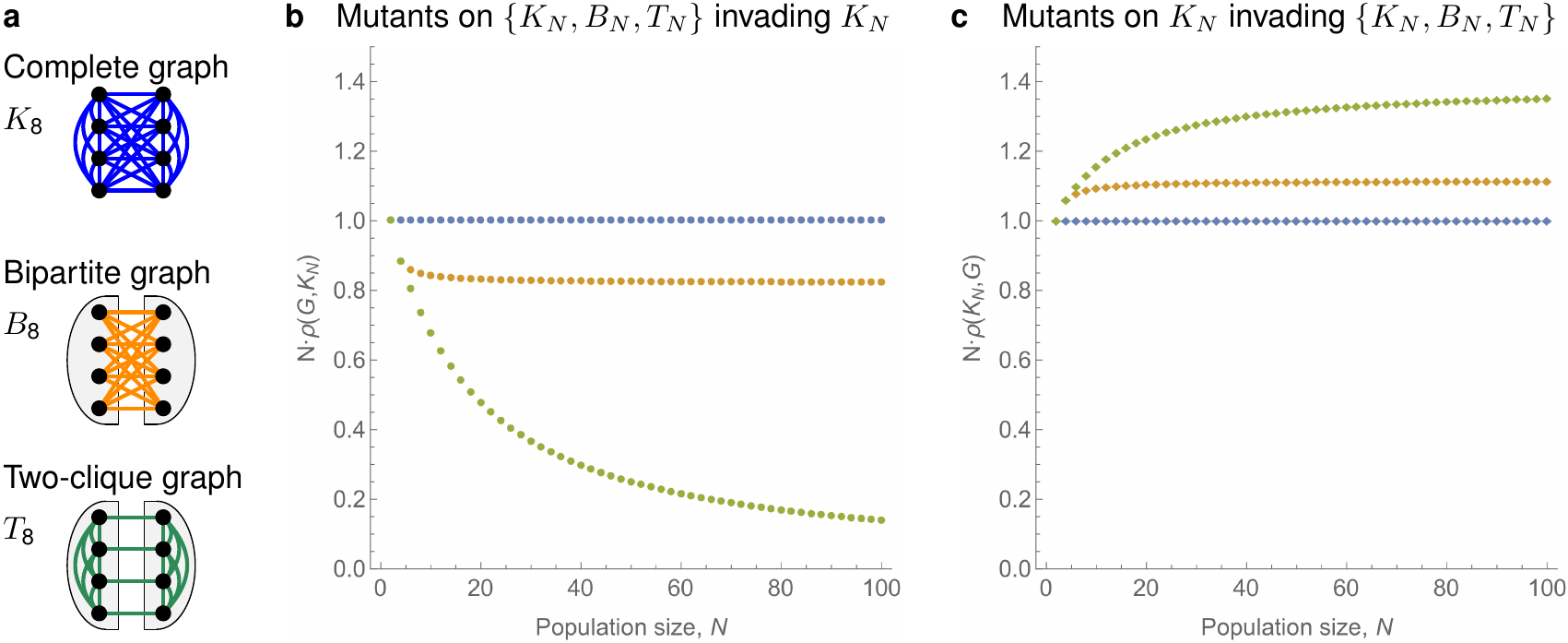}
        \caption{\textbf{Dense regular graphs.}
        \textbf{a,} In a (complete) Bipartite graph $B_N$ and a Two-clique graph $T_N$, each vertex is connected to $N/2$ other vertices (here $N$ is even).
        \textbf{b,} When the mutant lives on $B_N$, the fixation probability satisfies $\fp(B_N)\approx 0.82\cdot \frac1N$. In contrast, when the mutant lives on $T_N$, the fixation probability $\fp(T_N)$ tends to zero faster than $1/N$.
        \textbf{c,} When the residents live on $B_N$ or $T_N$, we have $\fpx(B_N)\approx 1.1\cdot \frac1N$ ans $\fpx(T_N)\approx 1.4\cdot\frac1N$.
        }
\label{fig:bntn}
\end{figure}

We prove that $\fp(B_N)>0.58/N$ (see \thmbipartite{} in \SI).
Since $\fp(K_N)=1/N$, this implies that missing roughly half of the edges only reduces the fixation probability by a constant factor, independent of the population size $N$.
In fact, numerical computation shows that $N\cdot \fp(B_N)\approx 0.82$ whereas for the Two-clique graph we observe $N\cdot \fp(T_N)\to 0$, see~\cref{fig:bntn}b.

The intuition for this distinction is as follows.
On both graphs, the state of the system at any given time point is completely described by the frequencies $N_L\in[0,N]$ and $N_R\in[0,N]$ of mutants in the left and the right half.
On $B_N$, the two frequencies remain roughly equal throughout the process ($N_L\approx N_R$):
indeed, once say $N_L\gg N_R$, more mutant offspring is produced on the left and they migrate to the right, thereby helping balance the numbers again. 
In contrast, on $T_N$ the mutants migrate rarely, thus the lineage produced by the initial mutant remains trapped in one half for substantial amount of time. 
Throughout that time, the mutants are ``blocking'' each other from spreading more than they would block each other if they were split evenly between the two halves:
indeed, with all mutants in one half, the probability that a reproducing mutant replaces another mutant (thus not increasing the size of the mutant subpopulation) is twice as large, as compared to the situation where the mutants are evenly split.
For small mutant subpopulations, this effect is non-negligible and it causes the fixation probability $\fp(B_N$) to decay faster than inversely proportionally to $N$.

Regarding $\fpx$, we observe $N\cdot \fpx(B_N)\approx 1.11$ and $N\cdot\fpx(T_N)\approx 1.4$, see~\cref{fig:bntn}c.
The intuition is that when mutants live on a Complete graph $K_N$, the offspring is equally likely to migrate to any location.
By randomness, the condition $N_L\approx N_R$ is thus maintained throughout most of the early stages of the process.
Therefore, as with $\fp(B_N)$, both $\fpx(B_N)$ and $\fpx(T_N)$ are inversely proportional to $N$.
To sum up, the graphs $B_N$ and $T_N$ show a considerably different behavior in terms of $\fp$ but a qualitatively comparable behavior in terms of $\fpx$.

%
%
\subsection{Lattice graphs}\label{sec:lattices}
Here we study sparse regular graphs, specifically lattice graphs.
Lattices exist in any number of dimensions.
We focus on one- and two-dimensional lattices, since those are biologically relevant.
For each dimension, we study the effect of increased or decreased connectivity (degree) of the lattice on the fixation probability of an invading mutant.

%
%
\paragraph{One-dimensional lattices.}
In one dimension, we consider circulation graphs $\Cir^d_N$ (already studied in this context from a different point of view, see~\cite{spirakis2021extension}).
For a fixed even integer $d$, a \textit{$d$-Circulation} graph, denoted $\Cir^d_N$, consists of $N$ vertices arranged in a cycle, where each vertex is connected to $d$ other vertices, namely the next $d/2$ vertices and the previous $d/2$ vertices in the cyclic order, see~\cref{fig:1d-lattices}a.

To shorten the notation, we denote by $\fponed(d_1,d_2)= \fp(\Cir^{d_1}_N,\Cir^{d_2}_N)$ the fixation probability of a mutant living on a one-dimensional lattice $\Cir^{d_1}_N$ with degree $d_1$ versus a population of residents living on a one-dimensional lattice $\Cir^{d_2}_N$ with degree $d_2$.
Note that when $d_1=d_2=d$ then $\fponed(d,d)=1/N$.

\begin{figure}[h] 
	\centering
\includegraphics[width=\linewidth]{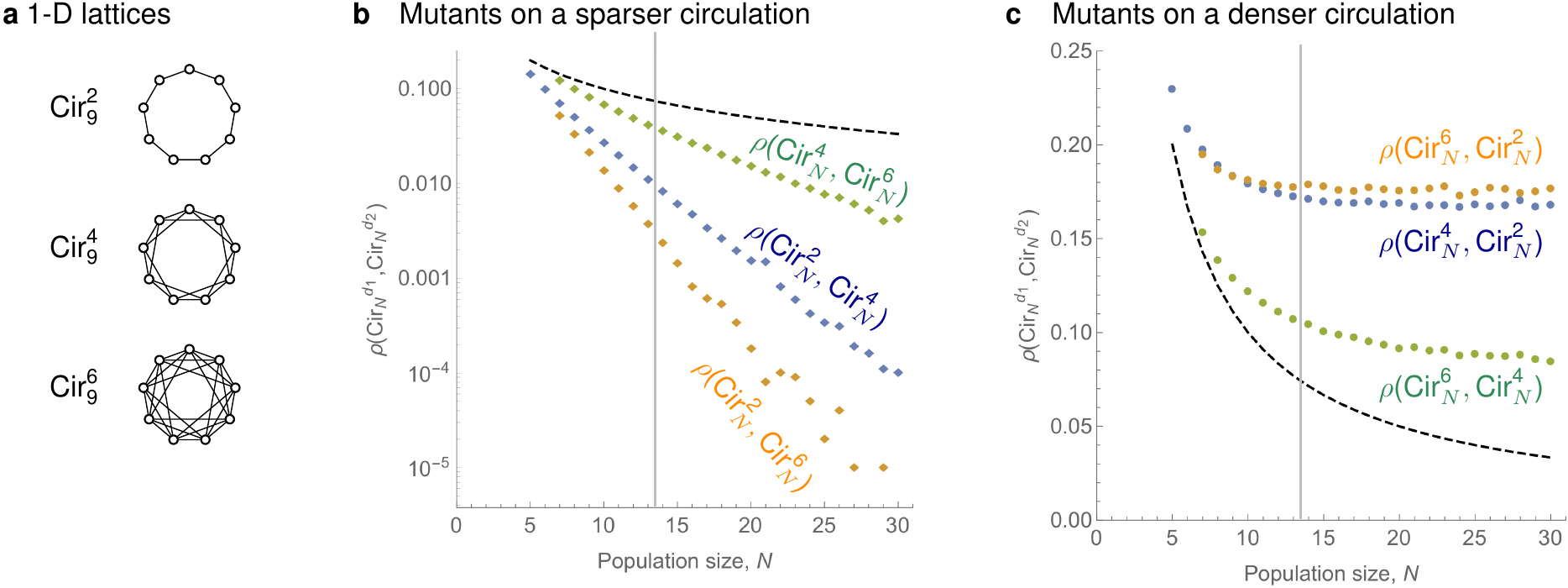}
        \caption{\textbf{Overlaying 1-D lattices with different connectivities.}
        \textbf{a,} A circulation graph $\Cir^d_N$ is a 1-dimensional lattice with periodic boundary and connectivity (degree) $d$.
        We consider $d\in\{2,4,6\}$.
        \textbf{b,} When mutants live on a less connected graph ($d_1<d_2$), their fixation probability decays to 0 at an exponential rate as $N\to\infty$ (here $y$-axis is log-scale).
        \textbf{c,} In contrast, when mutants live on a more densely connected graph ($d_1>d_2$), their fixation probability tends to a constant. 
        In both panels, the black dashed line shows the neutral baseline $1/N$.
        The values for $N\le 13$ are computed by numerically solving a large system of linear equations.
        The values for $N\ge 14$ are obtained by simulating the process $10^5$ times and reporting the proportion of the runs that terminated with the mutant fixating.
        }
        \label{fig:1d-lattices}
\end{figure}

When the degrees $d_1$, $d_2$ of the mutant and resident graph differ, the fixation probability crucially depends on which of the two degrees is larger. 
When the mutant graph has a lower connectivity ($d_1<d_2$) then $\fponed(d_1,d_2)$
 tends to 0 exponentially quickly as $N\to\infty$, see~\cref{fig:1d-lattices}b.
In contrast, when the mutant graph has a higher connectivity ($d_1>d_2$) then $\fponed(d_1,d_2)$
 tends to a positive constant $c$ that depends on $d_1$ and $d_2$, see~\cref{fig:1d-lattices}c.
Specifically, for large $N$ we observe that  $\fponed(4,2)\approx 0.16$, $\fponed(6,2)\approx 0.17$ and $\fponed(6,4)\approx 0.09$.

Those results are in agreement with bounds $0.11\le \fponed(4,2)\le 0.25$ that we prove analytically by a stochastic domination argument (see~\thmoned{} in \SI).
The intuition behind the argument is that once the mutants form a contiguous block of a large size, the block is more likely to expand rather than to diminish at both interfaces.
Indeed, the probability of gaining the boundary node is the same as losing the (other) boundary node but, on top of that, mutants could skip the boundary node, invade the interior of the resident territory and only after that gain the skipped node.
This event has a non-negligible probability of happening, hence there is a positive bias favoring the spread of mutants.
For a formal proof, see~\thmoned{} in \SI.

%
%
\paragraph{Two-dimensional lattices.}
 In two dimensions, we consider graphs drawn on a square lattice with periodic boundary condition.
For instance, by connecting each vertex to its 4 closest vertices (Von Neumann neighborhood), we obtain a graph $\Sq^4_N$, see~\cref{fig:2d-lattices}a.
Similarly, by connecting to 8 closest vertices (Moore neighborhood) we obtain a graph $\Sq^8_N$.
We also consider other graphs $\Sq^d_N$ with different connectivities $d\in\{6,12,20\}$.
We again shorten the notation by denoting $\fptwod(d_1,d_2)=\fp(\Sq^{d_1}_N,\Sq^{d_2}_N)$.

\begin{figure}[h] 
	\centering
\includegraphics[width=\linewidth]{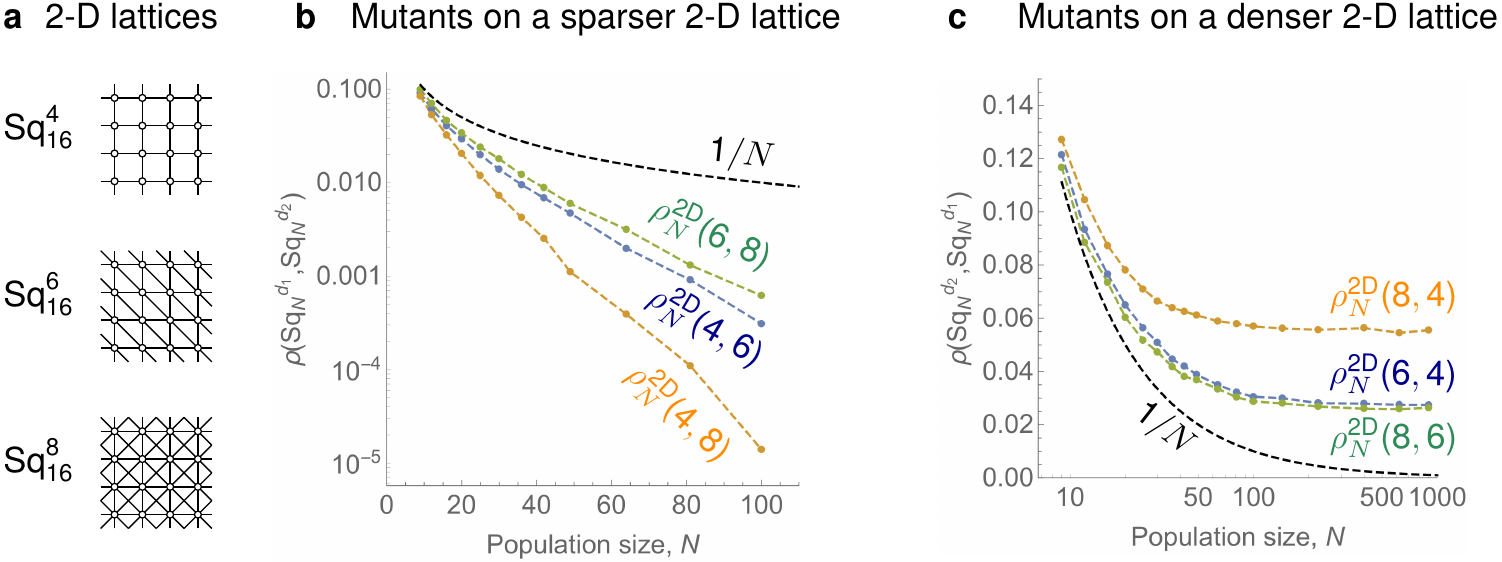}
        \caption{\textbf{Overlaying 2-D lattices with different connectivities.}
        \textbf{a,} We consider two-dimensional lattices with degree $4$ (Vonn Neumann neighborhood), $6$ (triangular grid), and $8$ (Moore neighborhood), and with dimensions $3\times 3,3\times 4,\dots,30\times 30$.
        \textbf{b, c} Similarly to the 1-D case, the fixation probability decays to 0 exponentially quickly when $d_1<d_2$, whereas it tends to a positive constant when $d_1>d_2$.
 The black dashed line shows the baseline $1/N$.
 The values are obtained by simulating the process (at least $10^5$ repetitions per data point).
        }
        \label{fig:2d-lattices}
\end{figure}

The results are analogous to the case of one-dimensional lattices.
When the mutants live on a less connected lattice,
their fixation probability tends to 0 exponentially quickly.
In contrast, when they live on a more densely connected lattice,
their fixation probability tends to a constant as the population size $N$ tends to infinity (see~\cref{fig:2d-lattices}).

%
%
\paragraph{Effective fitness.}
The behavior of the fixation probability for pairs of low-dimensional lattices is reminiscent of the behavior of the fixation probability
$\fpknr$ of a single mutant with relative reproductive rate $r\ne 1$ in a well-mixed population of $N-1$ other residents.
In that setting, we have $\fpknr=\frac{1-1/r}{1-1/r^N}$.
For any fixed $r\ne 1$, the formula exhibits one of two possible behaviors in the limit $N\to\infty$.
When $r<1$ then $\fpknr$ decays approximately as $1/r^N$.
In contrast, when $r>1$ then it tends to a positive constant $1-1/r$.
(When $r=1$ we have $\fpknr=1/N$ by symmetry.)

This suggests a possible interpretation:
for the neutral mutant, living on a more densely connected lattice has a comparable effect on the fixation probability as having a certain relative reproductive advantage $r_{d_1,d_2}$.
Formally, given a population size $N$ and two lattices $L_N$, $L'_N$ we define the \textit{effective fitness}, denoted
$r(L_N,L'_N)$, as the unique number $r$ such that 
\[ \fp(L_N,L'_N) = \fpknr.
\]
In other words, the effective fitness is such a number $r(L_N,L'_N)$, that a neutral mutant on a lattice $L_N$ invading a lattice $L'_N$ has the same fixation probability as a mutant with relative reproductive advantage $r(L_N,L'_N)$ in a well-mixed population.

For pairs of low-dimensional lattices with different connectivities $d_1,d_2$, the effective fitness can be computed from the data presented above, see~\cref{fig:effective}.
We observe that while the effective fitness depends on the connectivities $d$, $d'$ of the two lattices and on their dimensionality,
it is mostly independent of the population size $N$.

\begin{figure}[h] 
	\centering
\includegraphics[width=\linewidth]{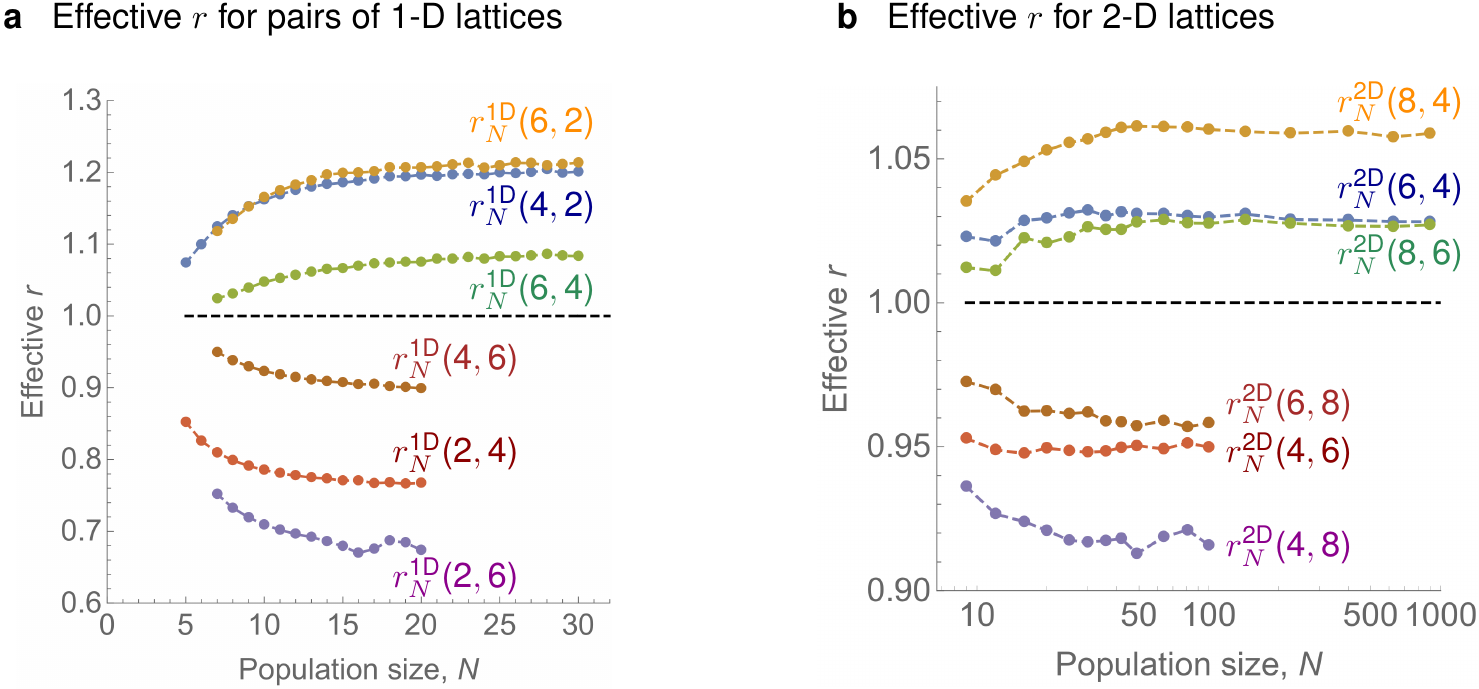}
        \caption{\textbf{Effective fitness.} Given the connectivities $d$, $d'$ of the mutant and resident lattice, we compute the effective fitness that would result in the same fixation probability, had both types lived on a Complete graph.
        \textbf{a,} One-dimensional lattices $\Cir^d_N$ with $d\in\{2,4,6\}$.
        \textbf{b,} Two-dimensional lattices $\Sq^d_N$ for $d\in\{4,6,8\}$. We have
        $\rtwod(8,4)\approx 1.06$,
        $\rtwod(6,4)\approx \rtwod(8,6)\approx 1.03$, and
        $\rtwod(6,8)\approx 0.96$,
        $\rtwod(4,6)\approx 0.95$,
        $\rtwod(4,8)\approx 0.92$.
        In both panels, the black dashed line shows the neutral baseline $r=1$.
        }
        \label{fig:effective}
\end{figure}

 \section{Discussion}
In this work, we studied the effect of mutations that, rather than altering the reproductive rate of the affected individual,
alter how the individual experiences the population structure.
To that end, we considered a powerful framework based on the classical Moran Birth-death process on graphs, in which
the two types of individuals (the novel mutant and the existing residents) perceive the population structure through different graphs.
As the key quantity, we studied the probability $\fp(G^A,G^B)$ that a single neutral mutant who perceives the population structure as a graph $G^A$ successfully invades the population of residents who perceive the population structure as a graph $G^B$.
For small population sizes, we computed the pairwise fixation probabilities numerically, and we observed that $\fp(G^A,G^B)$ tends to be higher when $G^A$ contains many edges (that is, the mutant is more motile) and when $G^A$ is regular.
We note that the latter aspect contrasts with other models of motility, where an increased dispersal potential of the mutant generally diminishes the fixation probability~\cite{thalhauser2010selection,manem2014spatial,krieger2017effects}.

Next, motivated by island models, we considered two regular graphs with the same total number of edges
and we showed that the corresponding fixation probabilities are asymptotically different.
In particular, as the population size~$N$ increases, the fixation probabilities decay at different rates.
Thus, in the asymptotic sense, the Isothermal Theorem of~\cite{lieberman2005evolutionary} is strongly violated.

Finally, we studied the biologically relevant case of 1- and 2-dimensional lattices
and we showed that the dispersal radius has similar effect on the fixation probability as the reproductive rate.
Recall that in large unstructured populations, a beneficial mutation fixates with constant probability, whereas the fixation probability of a deleterious mutation is exponentially small.
Likewise, neutral mutants on lattices with larger dispersal radius have a constant chance of successfully fixating, whereas
having lower dispersal radius leads to fixation of the mutant only with exponentially small probability.
Thus, in terms of the fixation probability of the mutant, perceiving the population through a more densely connected lattice
is effectively equivalent to having an increased reproductive rate.

Moving on to more complex (though perhaps less realistic) population structures, many natural questions arise.
We conclude by commenting on three of them.
Recall that for any graph $G_N$ on $N$ nodes we have $\fp(G_N,G_N)=1/N$~\cite{broom2010two}.

First, \cref{fig:n6} suggests that $\fp(G^A_N,K_N)< 1/N$ for all mutant graphs $G^A_N\ne K_N$.
While we can prove that $\fp(M_N,K_N)< 1/N$ for a graph $M_N$ that misses a single edge (see~\thmminusedge{} in \SI),
the general claim is left as an open problem.
Similarly, we do not know whether $\fp(K_N,G^B_N)>1/N$ holds for all resident graphs $G^B_N\ne K_N$ (we do know that it holds for $G^B_N=M_N$).

Second, following the game theory perspective, Melissourgos et al.~\cite{spirakis2021extension} asked what is the best mutant response to a given resident graph. That is, given a resident graph $G^B_N$ on $N$ nodes, which mutant graph $G^A_N$ on $N$ nodes maximizes the fixation probability $\fp(G^A_N,G^B_N)$?
Our results for small graphs show that although the Complete graph $K_N$ is frequently the best mutant response,
it is not always the case, see~\cref{fig:game}.
In particular, when the residents live on a Star graph $S_6$, the population is easier to invade through a graph $M_6$ that misses a single edge, rather than through the Complete graph $K_6$ -- direct computation gives
$\fp(M_6,S_6)>0.643>0.641>\fp(K_6,S_6)$.
We note that the difference is minor -- both mutant graphs $M_6$ and $K_6$ provide a fixation probability well over the neutral threshold value $1/6\approx 0.167$.

\begin{figure}[h] 
	\centering
\includegraphics[scale=1]{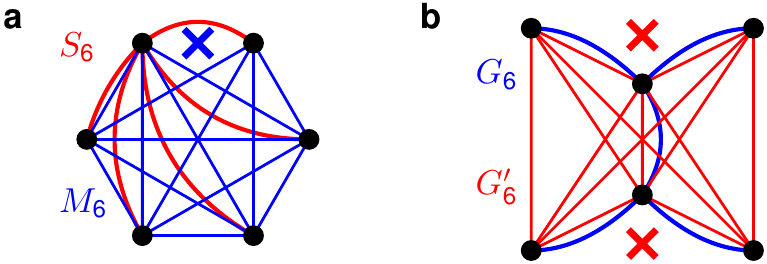}
        \caption{\textbf{Best-response graphs.} The Complete graph is sometimes not the best response when optimizing the fixation probability $\fp$.
        \textbf{a.} The resident population living on the Star graph $S_6$ (red) is easier to invade by mutants living on $M_6$ (blue) than by mutants living on the Complete graph $K_6$.
        \textbf{b,} The mutants living on a graph $G_6$ (blue) have a harder time invading the graph $G'_6$ than they have invading the Complete graph $K_6$.
        }
        \label{fig:game}
\end{figure}

For the complementary question of what is the best resident response $G^B_N$ to a given mutant graph $G^A_N$, the situation is analogous:
while the Complete graph is generally hard to invade, it is sometimes not the hardest one.
As an example (see~\cref{fig:game}b), when mutants live on a graph $G_6$ then for the graph $G'_6$ we have
$\fp(G_6,G'_6)<0.025<0.026<\fp(G_6,K_6)$.

Third, we observe that when mutants and residents live on different graphs $G^A_N\ne G^B_N$ with the same edge density, the fixation probability $\fp(G^A_N,G^B_N)$ typically drops below $1/N$.
The intuition is that the mutant subpopulation tends to form clusters in $G^A_N$ but not necessarily in $G^B_N$.
As a consequence, mutants block each other from spawning onto a resident 
but they do not guard each other from being replaced by residents. 

As an extreme example of this phenomenon, suppose that mutants live on a long cycle $G^A_N=C_N$ and they currently form a contiguous block of 10 individuals.
The probability $p^+$ that, in a single step, a mutant is selected for reproduction and its offspring replaces a resident is equal to $p^+=1/N$.
However, if the residents perceive the population as a completely different cycle $G^B_N=C'_N$ (such that no two mutants are adjacent in $C'_N$), then the probability $p^-$ that a resident replaces a mutant equals $p^-=10/N$.
Thus, in a single step, the size of the mutant subpopulation is $10\times$ more likely to decrease than it is to increase.
To some extent, similar effects occur whenever the two graphs $G^A_N$ and $G^B_N$ differ.
This suggests that for distinct graphs $G^A_N\ne G^B_N$ we typically have $\fp(G^A_N,G^B_N)<1/N$.

In one direction, this phenomenon can be easily overcome, for instance when one graph is denser than the other one.
Moreover, as witnessed by the two Star graphs depicted in~\cref{fig:n3}, there are pairs of irregular graphs for which the phenomenon is  overcome in both directions.
However, we are not aware of any such pair $G_N$, $G'_N$ of regular graphs.
Hence this is another open problem:
do there exist two regular graphs such that both $\fp(G_N,G'_N)>1/N$ and $\fp(G'_N,G_N)>1/N$?

\section*{Acknowledgements}
K.C. acknowledges support from ERC Consolidator grant no. (863818: ForM-SMart).

\section*{Author Contributions}
All authors designed the research.
J.T. and K.K. performed the mathematical analysis.
J.T. wrote the computer code and produced the figures.
All authors wrote the manuscript.

\section*{Competing interests}
The authors declare no competing interests.

\section*{Code and Data availability}
The datasets generated during and/or analysed during the current study are available in the Figshare repository,
\url{https://figshare.com/s/2d9cc41100151547b61a}.
 
\bibliographystyle{naturemag}
\bibliography{citations}

\input{si.tex}

\end{document}

%% file: si.tex

 \section*{Appendix}
In this appendix, we list the statements of our Theorems and their formal proofs.
In~\cref{sec:minusedge} we prove that the Complete graph is locally optimal 
 (\cref{thm:minusedge}).
In~\cref{sec:bipartite} we prove that for the (complete) Bipartite graph $B_N$ both 
 $\fp(B_N)$ and $\fpx(B_N)$ are within a constant factor of $1/N$ (\cref{thm:bipartite}).
Finally, in~\cref{sec:app-1d-lattices} we prove that for 1-dimensional lattices with dispersal radii 2 and 1,
the fixation probability $\fponed(4,2)$ is bounded away from 0 (\cref{thm:1d-lattices}).

 %
 %
\section{The Complete graph is locally optimal}\label{sec:minusedge}
Here we show that when it comes to $\fp(G_N)$ and $\fpx(G_N)$, the Complete graph is locally optimal.
That is, we show that the graph $M_N$ obtained from $K_N$ by removing a single edge satisfies
$\fp(M_N)=\fp(M_N,K_N)<1/N$ and $\fpx(M_N)=\fp(K_N,M_N)>1/N$.

\begin{theorem}[Complete graph is locally optimal]\label{thm:minusedge}
Fix $N\ge 2$ and let $M_N$ be a graph obtained from the Complete graph $K_N$ by removing a single edge.
Then
\[ \fp(M_N,K_N)=\frac{N-2}{(N-1)^2}<\frac1N \qquad\text{and}\qquad \fp(K_N,M_N)=\frac{1}{N-1}>\frac1N.
\]
\end{theorem}

\begin{proof} Fix $N\ge 2$. First we focus on $\fp(M_N)=\fp(M_N,K_N)$.
Denote by $\fii(a,b)$ the fixation probability starting from a configuration with 
$a$ mutants among the $N-2$ fully connected vertices 
and $b$ mutants among the other two vertices that miss one edge. 
We claim that
\[\fii(a,b)=\frac{(a+b)(n-2)+\frac12b(b-1)}{(n-1)^2}.
\]

Clearly, the formula satisfies $\fii(0,0)=0$ and $\fii(n-2,2)=1$. Therefore it suffices to check that given an arbitrary configuration (that is, any $0\le a\le n-2$ and $0\le b\le 2$), the expected value of the formula after a single transition doesn't change.

The transition probabilities are as follows:

\begin{enumerate}
\item[(i)] $p_{b+}\equiv\Pr[(a,b)\to (a,b+1)]=
\frac{a}{n}\cdot \frac{2-b}{n-1}$
\item[(ii)] $p_{a+}\equiv\Pr[(a,b)\to (a+1,b)]=
\frac{a}{n}\cdot \frac{n-2-a}{n-1} + \frac{b}{n}\cdot \frac{n-2-a}{n-2}$
\item[(iii)] $p_{b-}\equiv\Pr[(a,b)\to (a,b-1)]=
\frac{n-a-b}{n}\cdot \frac{b}{n-1}$
\item[(iv)] $p_{a-}\equiv\Pr[(a,b)\to (a-1,b)]=
\frac{n-a-b}{n}\cdot \frac{a}{n-1}$
\end{enumerate}

Ignoring the shared denominator $(n-1)^2$, the expected change in the value produced by the formula in the respective cases is as follows:
\begin{enumerate}
\item[(i)] $\Delta\fii_{b+}= +\,n-2+b$
\item[(ii)] $\Delta\fii_{a+}= +\,n-2$
\item[(iii)] $\Delta\fii_{b-}= -(n-2+b-1)$
\item[(iv)] $\Delta\fii_{a+}= -(n-2)$
\end{enumerate}

Denoting $I=\{b+,a+,b-,a-\}$ we compute

\begin{align*}
\Delta\fii&\equiv \sum_{i\in I} p_i\cdot\Delta\fii_i=
\frac1{n(n-1)}\cdot\Big[ a(2-b)(n-2+b) + a(n-2-a)(n-2) \\
&+ b(n-2-a)(n-1) 
- (n-a-b)b(n-2+b-1) -(n-a-b)a(n-2)
\Big]
\end{align*}
Grouping the terms with $n$ raised to the same power, the terms in the square brackets on the right-hand side can be rearranged as $n^2\cdot X+n\cdot Y+Z$, where $X=a+b-b-a=0$,

\begin{align*}
Y&= (2a-ab)+(-2a-a^2-2a) + (-2b-ab-b) \\
 &\hspace{1.87cm}- (-ab-b^2-2b+b^2-b) - (-2a-a^2-ab)\\
&= 0,
\end{align*}
and

\begin{align*}
Z&= (-ab^2+4ab-4a) +2a^2 + (ab+2b) +(b^3+ab^2-3b^2-3ab) -(2a^2+2ab)\\
&=b^3-3b^2+2b=b(b-1)(b-2).
\end{align*}
Since $Z=0$ for any admissible integer $b$ (recall $0\le b\le 2$), we are done.

 Regarding $\fpx(M_N)$, a completely analogous proof
 establishes a formula
\[\fii(a,b)=\frac{(a+b)(n-2)+b(3-b)/2}{(n-1)^2},
\]
which implies the desired
\[
\fpx(M_N)=\frac{1\cdot(n-2)+\frac{1\cdot 2}{2}}{(n-1)^2} = \frac1{n-1}.\qedhere
\]
\end{proof}

%
%
\section{Bipartite graph $B_N$ is within a constant factor of $K_N$}\label{sec:bipartite}
Recall that for $N$ even we denote by $B_N$ the (complete) Bipartite graph with parts of sizes $N/2$ and $N/2$.
We prove that for $N$ large, both $\fp(B_N)$ and $\fpx(B_N)$  are within a constant factor of $1/N$.
Recall that numerical experiments in the main text suggest that $\fp(B_N)\approx 0.82/N$ and $\fpx(B_N)\approx 1.1/N$, as $N\to\infty$.
\begin{theorem}[Bipartite graph $B_N$]\label{thm:bipartite} Fix $N\ge 2$. Then
\[
\fp(B_N)> \frac1{e-1}\cdot\frac1N \quad\text{and}\quad \fpx(B_N)<\frac e{e-1}\cdot\frac1N.
\]
\end{theorem}
\begin{proof}
First we bound $\fp(B_N)$.
Consider a time point at which there are $n$ mutants in total: $k$ of them in one part and $n-k$ in the other part. The probability $p^+(k,n-k)$ that in the next time step we gain a mutant equals
\begin{align*}
p^+(k,n-k) &=\frac kN\cdot\frac{N/2-(n-k)}{N/2} +\frac{n-k}N\cdot\frac{N/2-k}{N/2} \\
&= \frac{ k(N-2n+2k)+(n-k)(N-2k)}{N^2}\\
&=\frac{n\cdot N-4k(n-k)}{N^2}.
\end{align*}
Since $4k(n-k)\le n^2$ for any $k=0,\dots,n$ (and with equality only for $k=n/2$), we can bound
\[
p^+(k,n-k)\ge \frac{n\cdot N-n^2}{N^2}=\frac{n(N-n)}{N^2}.
\]
On the other hand, the probability $p^-(k,n-k)$ that in the next time step we lose a mutant equals
\[
p^-(k,n-k)=\frac{N-n}N\cdot\frac{n}{N-1}
\]
and thus
\[ \frac{p^-(k,n-k)}{p^+(k,n-k)}\le \frac N{N-1}\equiv t
\]
Since this expression is independent of $k$ and $n$, plugging it in the standard formula for the absorption probability of a 1-dimensional Markov chain we get
\[
\fp(B_N)\ge \frac1{1+ \sum_{i=i}^{N-1} t^i} =\frac{t-1}{t^N-1}.
\]
Since $t-1=1/(N-1)>1/N$ and
\[
t^N=\left(\frac N{N-1}\right)^N=\left(1+\frac 1{N-1}\right)^N\ge e,
\]
 where $e\doteq2.718$ is the famous Euler's constant, we thus obtain
\[
\fp(B_N)> \frac1{e-1}\cdot \frac1N>\frac{0.58}N.
\]

Regarding $\fpx(B_N)$, a completely analogous proof yields

\begin{align*}
p^+(k,n-k) &= \frac{n}N\cdot\frac{N-n}{N-1}\quad\text{and} \\
p^-(k,n-k) &= \frac{n\cdot N-4k(n-k)}{N^2}\ge \frac{n(N-n)}{N^2},
\end{align*}
Therefore
\[
\frac{p^-(k,n-k)}{p^+(k,n-k)}\ge \frac {N-1}{N} \equiv t.
\]
Finally, since $t^N=(1-1/N)^N\le 1/e$, we get
\[\fpx(B_N) \le \frac{t-1}{t^N-1} \le\frac{1/N}{1-1/e} = \frac e{e-1}\cdot\frac1N.\qedhere
\]
\end{proof}

%
%
\section{1-D Lattices}\label{sec:app-1d-lattices}
Recall that for a fixed even integer $d$ we denote by $\Cir^d_N$ the graph whose $N$ vertices, labelled $1,\dots,N$, are arranged along a circle and each vertex is connected with $d/2$ closest vertices clockwise and $d/2$ closest vertices counter-clockwise.
Also, recall that given two connectivities $d_1$, $d_2$ we shorthand
\[\fponed(d_1,d_2)=\fp(\Cir^{d_1}_N,\Cir^{d_2}_N).
\]
Figure~6 
from the main text suggests that when $d_1>d_2$ then the expression $\fponed(d_1,d_2)$ remains bounded away from 0 as $N\to\infty$,
Below we prove this in the special case $(d_1,d_2)=(2,1)$.
\begin{theorem}[1-D lattices]\label{thm:1d-lattices} We have
\[ 0.138 < \lim_{N\to\infty} \fponed(2,1) < 0.34.
\] 
\end{theorem}
\begin{proof}
By a \textit{configuration} we mean a subset of vertices occupied by mutants.
We denote the possible configurations as a sequence of numbers (corresponding to blocks of consecutive mutants) and symbols ``$\circ$'' (corresponding to individual residents).
That is, for instance the notation $k\circ 1$ denotes a configuration with $k$ consecutive mutants, then one resident, then one more mutant (and residents before and after).

\smallskip\noindent\textit{Proof of the upper bound.\ }
This is straightforward. We say that a step of the Moran process is \textit{active} if it changes the configuration.
Given a single mutant, there are two possible active steps leading to immediate extinction (each occurs with probability equal to $\frac1N\cdot \frac12$). On the other hand, there are four possible active steps where mutants reproduce (each occurs with probability equal to $\frac1N\cdot\frac14$). Thus, in total, with probability $\frac{2/2}{2/2+4/4}=\frac12$ the first active step results in the mutant extinction, hence $\fponed(2,1)\le 1/2$.

Accounting for trajectories that never reach a configuration with more than $m$ mutants, we can push this upper bound lower:
For instance, taking $m=2$, denote by $x$, $y$, $z$ the extinction probabilities from configurations 1, 2, $1\circ1$, respectively.
Conditioning on the first active step (see~\cref{fig:circulations-proof}a) we obtain:
\[ x=\frac14(2\cdot 1+y+z),\quad y\ge \frac15(2x+3\cdot 0), \quad z\ge \frac17(4x+3\cdot 0),
\]
hence $x\ge \frac14(2+\frac25x + \frac47x)$. This rewrites as $x\ge 35/53$, hence $\fponed(2,1)\le 18/53=0.34$ for any $N\ge 5$.

\smallskip\noindent\textit{Proof of the lower bound.\ }
For $k\ge 1$ denote by $a_k$ the fixation probability from configuration $k$. 
Similarly, for $k\ge 2$ denote by $b_k$ the fixation probability from the configuration $k\circ 1$.
Then $a_0=0$ and
\[a_1=\frac{ \frac22a_0+\frac24a_2+\frac24b_2}{\frac22+\frac24+\frac24} = \frac14(a_2+b_2).
\]

\begin{figure}[h] 
	\centering
	\includegraphics[width=\linewidth]{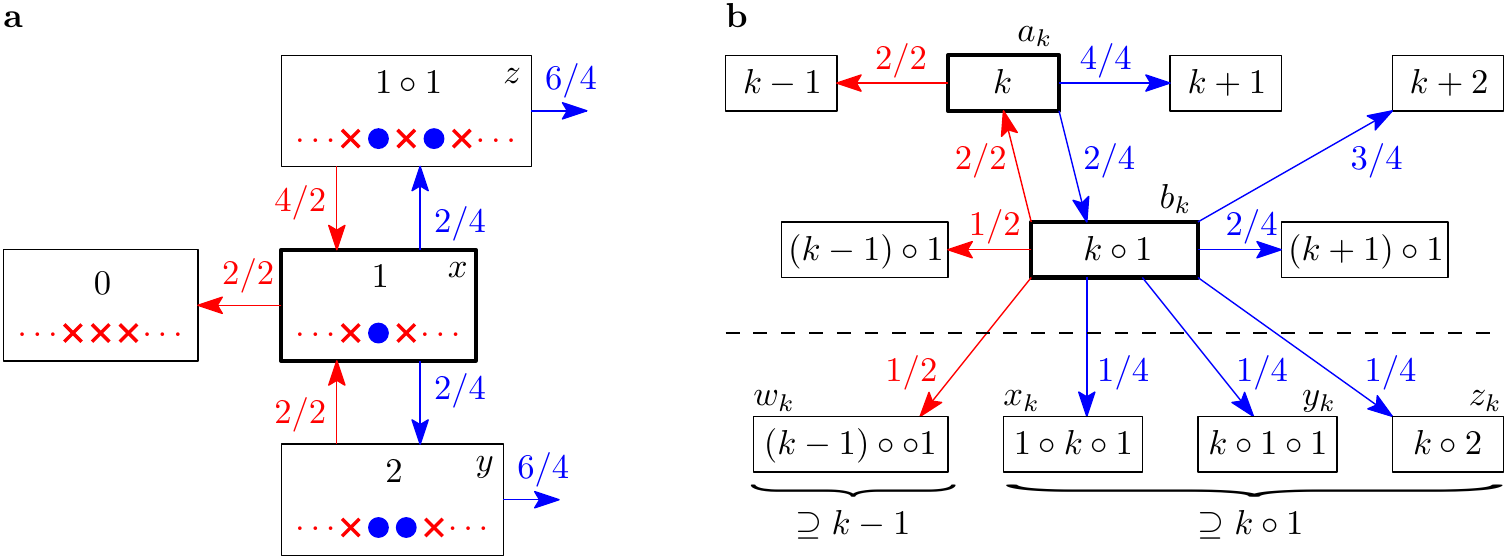}
        \caption{\textbf{Proof of~\cref{thm:1d-lattices}.}
        \textbf{a,}~Upper bound.
             Boxes represent configurations.
        Mutants are shown as blue disks, residents as red crosses.
        Blue (red) transitions correspond to mutants (residents) reproducing.
        A fraction $a/b$ denotes that there exist $a$ relevant edges, each pointing from a vertex with degree $b$.
         With a constant probability, the random evolutionary trajectory leads to mutant extinction without ever reaching a configuration with 3 or more mutants.
        \textbf{b,}~Lower bound.
         We consider a process $\Mlow$ where, any time a configuration below the dashed line is reached, we remove a mutant to instead reach one of the configurations above the dashed line.
         We then compute the fixation probability in $\Mlow$, which is a lower bound.
}
\label{fig:circulations-proof}
\end{figure}

For $k\ge 2$ we have (see~\cref{fig:circulations-proof}b)
\[a_k = \frac{\frac22 a_{k-1} + \frac24b_k + \frac44 a_{k+1}}{\frac22+\frac24+\frac44} = \frac15(2a_{k-1}+b_k+2a_{k+1}).
\]

For analogous expressions of $b_k$, we would need to introduce new variables $m_k$, $x_k$, $y_k$, $z_k$ to describe fixation probabilities from configurations listed below the dashed line.
Instead of computing the fixation probability exactly, we bound it by considering a different processes $\Mlow$, where mutants have lower fixation probability then in the original process $M$.
We define $\Mlow$ as follows:
The two processes coincide, except that when $M$ reaches any configuration below the dashed line,
we remove several mutants so as to obtain a configuration $a_i$ or $b_i$ (for some $i$), see~\cref{fig:circulations-proof}b.
For instance, when the current configuration is $k\circ1$, the resident from the gap is spawning an offspring, and the offspring moves one place left resulting in a configuration $(k-1)\circ\circ 1$, we additionally remove the single separated mutant, reaching the configuration $k-1$, where the mutants have fixation probability $a_{k-1}$.
Since for any two configurations $C\subseteq C'$ we have $\fp(C)\le \fp(C')$, the fixation probability in $\Mlow$ is indeed decreased, as compared to $M$.
(We note that by considering a process $\Mhigh$ where we occasionally add mutants so as to only visit configurations $a_i$, $b_i$, one can obtain a stronger upper bound than the one presented above.)

In $\Mlow$ we thus obtain

\begin{align*}
b_k&= \frac{ \frac22a_k + \frac12b_{k-1}+\frac12a_{k-1}  +  \frac34 b_k+\frac24b_{k+1}+\frac34a_{k+2} }{\frac22+\frac12+\frac12+\frac34+\frac24+\frac34} \\
&=\frac1{16}(2a_{k-1}+4a_k+3a_{k+2}  +  2b_{k-1}+3b_k+2b_{k+1})
\end{align*}

It remains to compute the fixation probability in $\Mlow$, in the limit $N\to\infty$. We do this by a standard argument.
Consider a ``potential function'' $\varphi$ which assigns a positive real number to each configuration, defined by $\varphi(k)=\alpha^k$, $\varphi(k\circ1)=c\cdot\alpha^k$, where $\alpha,c>0$ are positive real numbers (to be defined later).
By solving a system of 2 equations

\begin{align*}
5 &= 2/a + c + 2a \\
16c &=2/a+4+3a^2 + 2c/a+3c+2ca
\end{align*}
we find values $\alpha\doteq0.860$, $c\doteq0.954$ (the roots of certain degree-3 polynomials) for which the (expected) potential does not change in one step of the process
(that is, the function $\varphi$ is a martingale), except when at configuration $1$.
Now for any initial configuration $x\ne 1$, run the process starting from $x$ until it reaches either a configuration $1$ or the configuration $N$, and let $p_x$ be the probability that the former happens. Then
\[
 \varphi(x) = \E[\varphi(x)\mid \text{when the process ends}] = p_x\cdot\varphi(1) + (1-p_x)\cdot\varphi(N),
\]
which rewrites as
\[ p_x = \frac{\varphi(x)-\varphi(N)}{\varphi(1)-\varphi(N)}.
\]
Since $\alpha<1$ we have $\varphi(N)\to_{N\to\infty}0$, thus for $x\in\{2,1\circ1\}$ we can write
\[p_{2} \to_{N\to\infty} \frac{\alpha^2}{\alpha}=\alpha \quad\text{and}\quad p_{1\circ1} \to_{N\to\infty} \frac{c\cdot\alpha}{\alpha}=c.
\]
Now using the expressions
\[a_2 = p_2\cdot a_1 + (1-p_2)\cdot 1 \quad\text{and}\quad b_2 = p_{1\circ1}\cdot a_1 + (1-p_{1\circ1})\cdot 1.
\]
and plugging them into $a_1=\frac14(a_2+b_2)$ we finally obtain
\[ a_1 \to_{N\to\infty} \frac14\big(  (2-p_{2}-p_{1\circ1})  +  (p_{2}+p_{1\circ1})  \cdot a_1\big),
\]
hence
\[a_1 \to_{N\to\infty} \frac{2-p_{2}-p_{1\circ1} }{4-p_{2}-p_{1\circ1}}  =\frac{2-\alpha-c}{4-\alpha-c}>0.138. \qedhere
\]
\end{proof}